\documentclass[letterpaper]{article}

\usepackage[ruled,vlined]{algorithm2e}
\usepackage[square]{natbib}
\usepackage{pgfplots}

\usepackage{times}
\usepackage{mathrsfs}
\usepackage{graphicx}
\usepackage{amsmath}
\usepackage{amssymb}
\usepackage{amsthm}
\usepackage{multirow} 
\usepackage{url}
\usepackage{subcaption}
\usepackage{flushend}
\usepackage{tikz}
\usetikzlibrary{arrows,decorations,decorations.shapes,decorations.markings,decorations.pathreplacing,backgrounds,shapes,petri,topaths}
\usepackage{tkz-berge}

\newtheorem{theorem}{Theorem}
\newtheorem{corollary}[theorem]{Corollary}
\newtheorem{lemma}[theorem]{Lemma}
\newtheorem{proposition}[theorem]{Proposition}

\DeclareMathOperator{\argmin}{argmin}
\DeclareMathOperator{\argmax}{argmax}
\newenvironment{rtheorem}[1]{\medskip\noindent\textbf{Theorem~\ref{#1}.}\begin{itshape}}{\end{itshape}}
\newenvironment{rlemma}[1]{\medskip\noindent\textbf{Lemma~\ref{#1}.}\begin{itshape}}{\end{itshape}}
\newenvironment{rproposition}[1]{\medskip\noindent\textbf{Proposition~\ref{#1}.}\begin{itshape}}{\end{itshape}}

\def\shortcite{\citeyear}

\newcommand{\eps}{\epsilon}
\renewcommand{\vec}[1]{\mathbf{#1}}

\newcommand{\coursename}{(67686) Mathematical Foundations of AI}

\newcommand{\handout}[5]{
   \renewcommand{\thepage}{#1-\arabic{page}}
   \noindent
   \begin{center}
   \framebox{
      \vbox{
    \hbox to 5.78in { {\bf \coursename}
         \hfill #2 }
       \vspace{4mm}
       \hbox to 5.78in { {\Large \hfill #5  \hfill} }
       \vspace{2mm}
       \hbox to 5.78in { {\it #3 \hfill #4} }
      }
   }
   \end{center}
   \vspace*{4mm}
}

\newcommand{\labeq}[2]{
\begin{equation}
\label{eq:#1} #2
\end{equation}}
\def\ol{\overline}

\def\RPoA{\mbox{R-PoA}}

\def\PoA{\mbox{PoA}}

\newenvironment{proof-sketch}{\noindent{\bf Sketch of Proof}\hspace*{1em}}{\qed\bigskip}
\newenvironment{proof-idea}{\noindent{\bf Proof Idea}\hspace*{1em}}{\qed\bigskip}
\newenvironment{proof-of-lemma}[1]{\noindent{\bf Proof of Lemma #1}\hspace*{1em}}{\qed\bigskip}
\newenvironment{proof-attempt}{\noindent{\bf Proof Attempt}\hspace*{1em}}{\qed\bigskip}

\makeatletter
\def\fnum@figure{{\bf Figure \thefigure}}
\def\fnum@table{{\bf Table \thetable}}
\long\def\@mycaption#1[#2]#3{\addcontentsline{\csname
  ext@#1\endcsname}{#1}{\protect\numberline{\csname
  the#1\endcsname}{\ignorespaces #2}}\par
  \begingroup
    \@parboxrestore
    \small
    \@makecaption{\csname fnum@#1\endcsname}{\ignorespaces #3}\par
  \endgroup}
\def\mycaption{\refstepcounter\@captype \@dblarg{\@mycaption\@captype}}
\makeatother

\newcommand{\mathify}[1]{\ifmmode{#1}\else\mbox{$#1$}\fi}
\newcommand{\bigO}O

\newcommand\tup[1]{\left\langle #1 \right\rangle}

\renewcommand{\vec}[1]{{\mathbf #1}}

\newcommand{\remove}[1]{{}}

\def\GG{\mathcal{G}}

\newcommand{\xqed}{\mbox{\raggedright $\Diamond$}}

\newcommand{\step}[1]{\stackrel{{\scriptscriptstyle{#1}}}{\rightarrow}}

\newcommand{\newpar}[1]{
\vspace{-0mm}
\paragraph{#1}}
\newcommand{\subsec}[1]{
\vspace{-0mm}
\subsection{#1}}

\newcommand{\omittext}[1]{}

\def\CPoA{\mathrm{C-PoA}}
\newcount\Comments  
\definecolor{darkgreen}{rgb}{0,0.6,0}
\newcommand{\kibitz}[2]{\ifnum\Comments=1{\color{#1}{#2}}\fi}
\newcommand{\rmr}[1]{\kibitz{blue}{[RESHEF:#1]}}
\newcommand{\dcp}[1]{\kibitz{red}{[DAVID:#1]}}

\def\ol{\overline}
\def\os{\ol{\vec s}}

\begin{document}

\title{Congestion Games with Distance-Based Strict Uncertainty}
\author{Reshef Meir and David Parkes\thanks{Parkes is supported in this work by the TomKat Charitable Trust.}\\ Harvard University}

\maketitle

\begin{abstract}
We put forward a new model of congestion games where agents have uncertainty over the routes used by other agents. We take a non-probabilistic approach, assuming that each agent knows that the number of agents using an edge is within a certain range. Given this 
uncertainty, we model agents who either minimize their worst-case cost (WCC) 
or their worst-case regret (WCR), and study implications on equilibrium 
existence, convergence through adaptive play, 
and efficiency. 
Under the WCC behavior the game reduces to a modified congestion game,  and 
welfare improves when agents have moderate uncertainty. Under WCR behavior the game is not, in general, a congestion game, but we show convergence and efficiency bounds for a simple class of games.
\end{abstract}

\section{Introduction}

Congestion games~\cite{Ros73} provide a good abstraction for a wide
spectrum of scenarios where self-interested agents contest for
resources, and can be conveniently analyzed using game-theoretic
tools.

Recently, more complex models of congestion games have been
suggested, taking into account the incomplete information
agents may have when making a decision (e.g.~Ashlagi et al.~\shortcite{ashlagi2009two} and Piliouras et al.~\shortcite{piliouras2013risk}; see Related Work). Uncertainty may stem from
multiple sources, including uncertainty about the state of nature-- 
and thus the cost of resources --or about
other agents' actions.
We can imagine commuters choosing routes home from
work and facing uncertainty about road conditions (e.g., weather,
roadworks) as well as about the routes selected by others.
%

Rather than model uncertainty through a distributional model, 
%
%
we adopt a non-probabilistic approach of 
{\em strict uncertainty}. Indeed, 
extensive experimental and empirical studies have demonstrated that people have
difficulties in representing probabilities, and often adopt
other heuristics in place of 
probabilistic reasoning~\cite{tversky1974judgment,slovic1980facts}.
%
With strict uncertainty, each player 
faces a set of 
\emph{possible states}, where the cost of each action depends on the
(unknown) actual state, and decisions take into account risk attitudes and other biases,  often applying heuristics rather than optimization. Such alternative approaches to decision making in general, and to uncertainty in particular, have deep roots in the AI literature, largely due to the works of Herbert Simon on \emph{bounded rationality} and \emph{procedural rationality}~\cite{simon57,simon1987bounded}.
 
Having adopted a non-probabilistic approach, we must make two crucial
modeling decisions. 
First,
we must decide how each agent acts in the face
of strict uncertainty. 
The simplest behavior follows a {\em minimax approach}~\cite{wald1939contributions,simon57}, and assumes that the decision maker is trying to minimize her \emph{worst-case cost} (WCC). 
%
Another approach seeks to minimize the \emph{worst-case regret} (WCR) of the decision maker, which goes back to Savage~\shortcite{savage1951theory}, and has been also applied to games~\cite{hyafil2004regret}. Both cost measures are worst-case approaches, and suitable as an abstraction for the behavior of a rational but risk-averse agent.

Second, we need to determine which states are considered
possible by the agents.
To construct the set of possible states, we
adopt the recent model of \emph{distance-based
uncertainty}~\cite{MLR14}.
%
All agents share the same belief about the current
``reference state" of
the network (i.e., the load on every edge in a routing game), which may be available for example from an
external source such as traffic reports,
or from 
an agent's previous experience. However
agents vary in the accuracy they attribute to the reference state. Each agent $i$ has an uncertainty parameter, $r_i$, which  
 reflects a belief that the 
actual load is within some distance $r_i$ of
the reference load.
A higher $r_i$ may reflect either that an agent is less informed about the true congestion, or, alternatively, that she is more risk-averse. 

From each heuristic (WCC or WCR) we can derive a natural equilibrium concept. Intuitively, every action profile induces a reference state $\vec s$, and we consider the heuristic best response of every agent to the set of possible states around $\vec s$. State $\vec s$ is an \emph{equilibrium} if every agent minimizes her worst case cost (or regret) by keeping her current action.


As a simple example, if in some profile $100$ players are using a resource, then agent $i$ believes  the actual load to be anywhere between $100/r_i$ and $100r_i$. If $r_i=1$ for all agents, we get the standard complete information model as a special case: both minimax cost and minimax regret collapse to simple cost minimization, and  our equilibrium notion coincides with Nash equilibrium.

\paragraph{Our contribution.}

We study equilibrium behavior in {\em nonatomic congestion games}, under
our strict, distance-based model of uncertainty.
For simplicity and concreteness, we focus our presentation on routing games, where resources are edges in a graph, and valid strategies are paths from source to target.\footnote{
Any congestion game is equivalent to a routing game.}
%
%
%
%
%
 With worst-case cost players, we show that the game reduces to a modified,
complete-information routing game with {\em player-specific costs}.  Further, if all agents have the same uncertainty level we get a potential game.

 We are 
interested in how equilibrium quality (measured by the price of
anarchy~\cite{roughgarden2004bounding}) is affected by introducing uncertainty.
For routing games with affine cost functions, we show that the price of anarchy (PoA) under uncertainty decreases gradually from $\frac43$ (without uncertainty) to $1$, and then climbs back up, proportionally to the amount of uncertainty. We also show that in a population of agents with different uncertainty levels, the PoA is bounded by the PoA of the worst possible uncertainty level in symmetric affine games on parallel networks.\footnote{In the AAAI'15 published version, the last result was stated for any affine game, which is incorrect. A counter example appears in \cite{brown2017studies}. We thank Philip Brown for pointing this issue out.} 

With worst-case regret players the induced game is no longer a congestion game. Yet, we show that for a simple class of games a weak potential function exists, and thus equilibrium existence and convergence results are 
available. We give some preliminary results on PoA bounds 
with worst-case regret players.
Due to space constraints most of our proofs are omitted, and are available in 
a separate Appendix.

\section{Preliminaries}


\newpar{Nonatomic routing games.}

Following~Roughgarden~\shortcite{roughgarden2003price} and
Roughgarden and Tardos~\shortcite{roughgarden2004bounding}, a {\em nonatomic routing game}
(NRG) is a tuple $\GG = \tup{G,\vec c, m,\vec u,\vec v,\vec n}$,  where 
\begin{itemize}
	\item $G=(V,E)$ is a directed graph;
	\item $\vec c = (c_e)_{e\in E}$, $c_e(t)\geq 0$ is the cost incurred when $t$ agents use edge $e$;
	\item $m\in \mathbb N$ is the number of agent types;
	\item $\vec u,\vec v\in V^m$, where $(u_i,v_i)$ are the source and target nodes of type $i$ agents;
	\item $\vec n\in \mathbb R_+^m$, where $n_i \in \mathbb R_+$ is the total mass of type~$i$ agents. 
$n=\sum_{i\leq m} n_i$ is the total mass of agents.
\end{itemize}

 We denote by  $A_i \subseteq 2^E$ the set of all directed paths between the pair of nodes $(u_i,v_i)$ in the graph. Thus $A_i$ is the the set of {\em actions} available to agents of type $i$. We denote by $A = \cup_i A_i$ the set of all directed source-target paths. We assume that the costs $c_e$ are non-decreasing, continuous and differentiable.

A NRG is \emph{symmetric} if all agents have the same source and target, i.e., $A_i=A$ for all $i$.
A symmetric NRG is a \emph{resource selection game} (RSG) if $G$ is a graph of parallel links. That is, if $A=E$ and the action of every agent is to select a single resource (edge).

\newpar{Game states.} A \emph{state} (or action profile) is a vector $\vec s \in \mathbb R_+^{|A| \times m}$,
where $s_{f,i}$ is the amount of agents of type $i$ that use path
$f\in A_i$. In a valid state, $\sum_{f\in A_i}s_{f,i} = n_i$. The
total traffic on path $f\in A$ is denoted by $s_f = \sum_{i=1}^m
s_{f,i}$. 
The \emph{load state} $\os\in \mathbb R_+^{|E|}$ is a vector of aggregated edge loads derived from state $\vec s$, where $\ol s_e = \sum_{f:e\in f}s_f$. 
This is the total traffic on edge $e\in E$ via  all paths going through $e$. We also allow load states that \emph{cannot be derived} from a valid state.

%
Note that all the information relevant to the costs in state $\vec s$ is specified in the load state $\os$:
all agents using a particular edge $e$ suffer a cost of $c_e(\ol s_e)$ in state $\vec s$, and the cost of using a path $f\in A_i$ is $c(f,\os) = \sum_{e\in f}c_e(\ol s_e)$. 
Thus except in settings where agents' types or exact strategies matter, we may use $\vec s$ and $\os$ interchangeably.
The \emph{social cost} in a profile $\vec s$ is
$$SC(\vec s)=\sum_{i=1}^m \sum_{f\in A_i} s_{f,i} c(f,\os) = \sum_{e\in E} \ol s_e c_e( \ol s_e).$$ 
The second equality is since we multiply the cost $c(t)$ by the mass $t$ (``number of agents'') who experience it.

\newpar{Equilibrium and potential.}

Without uncertainty, a state $\vec s$ for an NRG is an
\emph{equilibrium} if for every agent type $i$ and
actions
$f_1, f_2 \in A_i$ with $s_{i,f_1} > 0$, $c(f_1,\os) \leq
c(f_2,\os)$. That is, if no agent can switch to a path with a lower
cost.  This is the analogy of a Nash equilibrium in nonatomic games.

In nonatomic games, $\phi(\vec s)$ is a \emph{potential function}, if
any (infinitesimally small) rational move, i.e., a move that decreases
the cost of the moving agents, also lowers the potential.  $\phi(\vec
s)$ is a \emph{weak potential function} if at any state there is at
least one such move (although some rational moves may increase
$\phi$).
Any game with a potential function is \emph{acyclic}, in the sense that such ``infinitesimal best responses''  of self interested agents are guaranteed to converge to a local minimum of the potential function
(and an equilibrium). 
A game with a weak potential may have cycles, but from any state there is  some path of rational moves that leads to an equilibrium.

It is well known that NRGs  have a potential function, which is defined as (we omit the argument $\GG$ when it is clear from the context): $\phi(\vec s) = \phi(\GG,\vec s) = \sum_{e\in E}\int_{t=0}^{\ol s_e}c_e(t) dt.$
Furthermore, in a NRG every local minimum of the potential is also a global minimum;  all equilibria  have the same social cost; and in every equilibrium all agents of type~$i$ experience the same cost~\cite{aashtiani1981equilibria,milchtaich2000generic,roughgarden2004bounding}. 

\newpar{Affine routing games.}

In an {\em affine} NRG, all cost functions take the form of a linear function. That is, $c_e(t) = a_e t + b_e$ for some constants $a_e,b_e\geq 0$. In an affine game $\GG$, the social cost can be written as $SC(\GG,\vec s)= \sum_{e\in E} a_e (\ol s_e)^2 + b_e \ol s_e$; and the potential as $\phi(\GG,\vec s) = \sum_{e\in E} \frac12 a_e (\ol s_e)^2 + b_e \ol s_e$.
{\em Pigou's example} is the special case of an affine RSG with two
resources, where $c_1(t)=1$ and $c_2(t)$ is defined
with $b_2=0$. We will use variations of
this example throughout the paper, and denote by $\GG_P(a_2,n)$ the
instance where $c_2(t) = a_2 t$, and there is a mass of $n$ agents.

\newpar{Potential and social cost.}

The 
social cost of every NRG  can be written as the potential of a
suitably modified game. For this, let $\hat \GG$ be a modification of $\GG$,
where we replace every $c_e(t)$ with $\hat c_e(t) = c_e(t) + t c'_e(t)$.
Then, $\phi(\hat \GG,\vec s) = SC(\GG,\vec s)$ for all $\vec
s$~\cite{roughgarden2007routing}. 
 For an affine game, 
the modified cost function is $\hat c_e(t) = 2a_e t +
b_e$; and $\phi(\hat \GG,\vec s) = \sum_{e\in E} a_e (\ol s_e)^2 + b_e \ol s_e =
SC(\GG,\vec s)$.

\paragraph{The price of anarchy.} 

Let $EQ(\GG)$ be the set of equilibria in game $\GG$.  The
\emph{price
  of anarchy} (PoA) of a game is the ratio between the social cost in
the worst equilibrium in $EQ(\GG)$ and the optimal social cost.
Since all equilibria have the same cost
, we can write $\PoA(\GG) =\frac{SC(\vec s^*)}{SC(OPT)}$, where $\vec s^*$ is an arbitrary equilibrium of $\GG$. 
In affine NRGs, it`
is known that  $\PoA(\GG)\leq \frac{4}{3}$, and this bound is attained by $\GG_P(1,1)$~\cite{roughgarden2004bounding}. 


\section{Introducing uncertainty}

In our strict uncertainty model, there is an
underlying base game $\GG$, which is a NRG.
 However given an
 action
profile (state) $\vec s$, 
each agent believes that there is some set of possible states, and
selects her action based on worst-case assumptions.
%

To define this set of possible states, we augment the description of every agent type with an \emph{uncertainty parameter} $r_i\geq 1$, and denote  $\vec r = (r_i)_{i\leq m}$. The special case where all $r_i=r$ is called \emph{homogeneous uncertainty}. 
We adopt distance-based
uncertainty, so that 
in a given state $\vec s$ (where the actual load on edge $e$ is $\ol s_e$),  a type~$i$
agent believes that the load is anywhere in the range $[\ol s_e/r_i,
\ol s_e\cdot r_i]$. 
Consequently, the agent believes that
the cost she
will suffer from using resource $e$ is between $c_e(\ol s_e/r_i)$ and
$c_e(\ol s_e \cdot r_i)$. 
Agents apply 
this reasoning separately to each
resource, thus the load state $\os'$ is considered \emph{possible} in load state $\os$ by a type $i$
agent, if $\ol s'_e\in [\ol s_e/r_i, \ol s_e\cdot r_i]$ for all $e\in E$.\footnote{\label{fn:modal}In the language of modal logic, we say that the $\os'$ is
\emph{accessible} from $\os$ if the above holds. Our accessibility
relation is symmetric, but non-transitive. While transitivity is a
well accepted axiom in epistemic models~\cite{aumann1999interactive},
we argue that it does not make sense when there is a natural metric over states.} 

In other words, consider the distance metric $d(\os,\os')=\min\{x\geq 0 : \forall e\in E,  \ol s_e \geq
\frac{\ol s'_e}{1+x} \wedge \ol s'_e \geq \frac{\ol s_e}{1+x} \}$. Then
$S(\vec s,r_i)=S(\os, r_i) = \{\os'\in \mathbb R_+^{|E|} : d(\os,\os')\leq r_i - 1\}$
is the set of load states that
a type~$i$ agent believes possible given 
$\os$.
Note that $\os'$ may not correspond to any actual state $\vec s'$, e.g. the total load on all paths  may not sum up to total mass $n$, as an agent may not know exactly how many other agents participate.


\subsec{Behavior and equilibria}

\paragraph{Worst-case cost.}

Under the WCC model, each agent cares about the worst possible cost of each action. Thus for an agent of type $i$, the effective cost of choosing path (action) $f\in A_i$ in state $\vec s$ is $c^*_{i}(f,\os) = \max\{c(f,\os'): \os' \in S(\os,r_i)\}$.

Every NRG $\GG$ and uncertainty vector $\vec r$ induce a new nonatomic game $\GG^*(\vec r)$, where the cost functions are $c^*_{i}$.
That is, a type~$i$ agent playing so as to minimize her worst-case cost in $\GG$, behaves exactly like a ``rational'' type $i$ agent (minimizing exact cost) in $\GG^\ast(\vec r)$.
{\em A priori}, $\GG^*(\vec r)$ is not a NRG,  but it has a very similar structure.
For every action $f\in A_i$, according to the WCC model,

$$c^*_i(f,\os) = \sum_{e\in f} \max\{c_e(\ol s'_e): \os' \in S(\os,r_i)\} =  \sum_{e\in f} c_e (r_i \ol s_e).$$

Since $c_e(r_i t)$ can be written as a player-specific cost function, $c^*_{i,e}(t)$, we have that $\GG^*(\vec r)$ is a NRG with 
{\em player-specific} costs~\cite{milchtaich2005topological}, where
each player type can adopt a different cost function; 
 e.g., for affine games $c^*_{i,e}(t) = r_i a_e t+b_e$. 


\newpar{Worst-case regret.}

We get a different modified game, $\GG^{**}(\vec r)$, under the WCR model.
The {\em regret} (for a type $i$ agent) of playing action $f$ in state
$\vec s'$ is defined as $REG_i(f,\os') = c(f,\os') - \min_{f'\in
  A_i}c(f',\os')$. Given this, the cost $c^{**}_i(f,\os)$ in the modified
game, which is the
\emph{worst-case regret} a type~$i$ agent may suffer for playing $f$, is
 defined as:

$$c^{**}_i(f,\os) = \max\{REG_i(f,\os'): \os' \in S(\os,r_i)\}.$$

This cost function  $c^{**}_i(f,\os)$ does not have a natural decomposition to edge-wise costs, since regret depends also on the load on unused edges. 
An example of WCC and WCR costs in a simple 2-resource RSG appear in Figure~\ref{fig:WCR_example}.
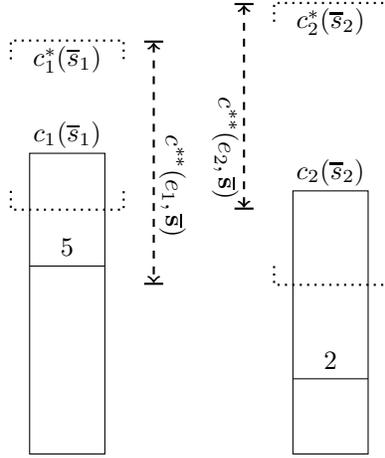
\begin{figure}
%
%
\begin{center}

\begin{tikzpicture}[scale=0.5]



\draw (0,0) rectangle (2,8);
\draw (0,5) -- (2,5);
\draw[dotted,thick] (-0.5,7) -- (-0.5,6.5) -- (2.5,6.5) -- (2.5,7);
\draw[dotted,thick] (-0.5,10.5) -- (-0.5,11) -- (2.5,11) -- (2.5,10.5);

\draw (7,0) rectangle (9,7);
\draw (7,2) -- (9,2);
\draw[dotted,thick] (6.5,5) -- (6.5,4.5) -- (9.5,4.5) -- (9.5,5);
\draw[dotted,thick] (6.5,11.5) -- (6.5,12) -- (9.5,12) -- (9.5,11.5);
\node at (1,5.5) {$5$};
\node at (8,2.5) {$2$};
\node at (1,8.5) {$c_1(\ol s_1)$};
\node at (8,7.5) {$c_2(\ol s_2)$};

\node at (1,10.5) {$c^*_1(\ol s_1)$};
\node at (8,11.5) {$c^*_2(\ol s_2)$};

\draw[|<->|,dashed,thick] (3.3,11) -- (3.3,4.5);
\draw[|<->|,dashed,thick] (5.7,12) -- (5.7,6.5);
\node[label=below:\rotatebox{-90}{$c^{**}(e_1,\os)$}] at (3.8,9) {};
\node[label=below:\rotatebox{-90}{$c^{**}(e_2,\os)$}] at (5.2,10) {};
%
%
%
%
%
%
%
\end{tikzpicture}
\end{center}
\caption{\label{fig:WCR_example}Two resources with base costs $c_1(\ol s_1) = 5+\ol s_1, c_2(\ol s_2)=2+\ol s_2$. The figure shows the true costs for $\ol s_1=3,\ol s_2=5$. The dotted brackets show the range of possible costs for $r=2$, where the upper bracket is the WCC cost. The dashed lines are the WCR costs. We can see that the better resource under WCC is $e_1$, but $e_2$ is better under WCR, as $c^{**}(e_2,\os) = 12- 6.5 = 11-4.5 = c^{**}(e_1,\os)$.}

 \end{figure}

\newpar{Equilibrium.}

A {\em WCC equilibrium} is a state where no agent can improve her worst-case cost w.r.t. her uncertainty level. By definition of the cost function $c^*$, the WCC equilibria of $\GG$ for uncertainty values $\vec r$ are exactly the Nash equilibria of $\GG^*(\vec r)$. 
Similarly , a {\em WCR equilibrium} is a Nash equilibrium of $\GG^{**}(\vec r)$.
Since both of $\GG^{*}(\vec r)$ and $\GG^{**}(\vec r)$ are special cases of nonatomic games, existence of equilibria follows from general existence theorems~\cite{schmeidler1973equilibrium}. However the other properties  of NRG, such as the existence of a potential function, and bounds on the PoA, are not guaranteed. 

\section{Routing Games with WCC players}

\paragraph{Equilibrium and convergence.}

 
For the special case of $r_i=r$ for all $i\leq m$, it is not hard to see that $\GG^*(r)$ is a non player-specific NRG. This is since $c^*_{e,i}(\os)= c_e(r \cdot \ol s_e)$  is only a function of $\ol s_e$. We denote this modified cost function  by $c^r_e(t)$.
It follows that $\GG^*(r)$ is a potential game, where $\phi(\GG^*(r),\vec s) = \sum_e \int_{t=0}^{\ol s_e}c^r_e(t) dt$. 
Thus $\GG^*(r)$ is acyclic, 
 the equilibria of $\GG^*(r)$ are the minima of $\phi(\GG^*(r),\vec s)$, and all equilibria have the same social cost.

The more interesting question is what properties of NRG are
maintained when agents have different uncertainty parameters.  We have
already noted that in $\GG^*(\vec r)$ there is at least one equilibrium. 
Player-specific RSGs are known to have a weak potential~\cite{milchtaich1996congestion}, but this does not preclude cycles.
Indeed, we show that a cycle may occur even in an RSG where agents only differ in their uncertainty level.
%
%
%
\begin{proposition}\label{th:example_cycles}
There is an RSG $\GG$ with 3 resources, and a vector $\vec r$ s.t. $\GG^*(\vec r)$ contains a cycle. 
\end{proposition}

\subsection{Equilibrium quality for affine games}

Recall that under the WCC model, agents play as if they take part of the game $\GG^*(\vec r)$, while their actual, realized costs are those
in underlying game $\GG$. 
We thus define the {\em Price of Anarchy for WCC players with uncertainty vector $\vec r$} as:

$$\CPoA(\GG,\vec r) = \max_{\vec s\in EQ(\GG^*(\vec r))} \frac{SC(\GG,\vec s)}{SC(\GG,OPT(\GG))}.$$

We  focus our analysis on games with affine costs, and 
look for bounds on $\CPoA(\GG,\vec r)$. In particular, we  explore   whether players with uncertainty reach better or worse social outcomes under  WCC behavior than under standard, complete information equilibria.
%

\paragraph{Homogeneous uncertainty.}

We start with the simplifying assumption that $r_i=r$ for all types.
Recall that in this case $\GG^*(r)$ is a non player specific NRG,
where the cost of each edge is modified to $c^r_e(t) = r a_e t +
b_e$.
These modified costs can be attained in other contexts that do not involve uncertainty. For example, this can
be achieved through taxation~\cite{cole2003much}. 
Following the discussion in Potential and Social Cost, 
 an optimal taxation scheme would perturb the cost functions so that the
 realized cost is $\hat c_e(t) = c_e(t)+t c'_e(t)$, as this will
 guarantee that $\phi(\hat \GG,\vec s)= SC(\GG,\vec s)$ for all $\vec s$.  In this way, minimizing the potential of $\hat \GG$, as happens
 in equilibrium, also minimizes the social cost in $\GG$ (see Section
 18.3.1 in Roughgarden~\shortcite{roughgarden2007routing} for a detailed explanation).

For the special case of affine games, it is easy to see that
the effect of uncertainty level $r=2$
is equivalent to that of an optimal taxation scheme. That is, $\phi(\GG^*(2),\vec s) =
SC(\GG,\vec s)$ for all $\vec s$.  This means that if all agents
 adopt the WCC viewpoint for an uncertainty type of $r=2$,
 then they would play the social optimum.  Unfortunately, the value of
 $r$ is not a design parameter---we cannot decide for the agents how
 uncertain they should be, since this reflects their beliefs.
 We would
 therefore like to have guarantees on equilibrium quality for any
 value of $r$. The next lemma provides our first result in this
 direction.

We denote $\phi^r(\vec s) = \phi(\GG^*(r),\vec s)$.
It will be convenient to treat the cases $r\geq 2$ and $r\leq 2$ separately.
	\begin{lemma}
\label{lemma:affine_optimal}
Let $\GG$ be an affine NRG, and suppose that $r\geq 2$. Then  for all $\os$,
	 $\phi^r(\vec s) \in [SC(\GG,\vec s),\frac{r}{2}SC(\GG,\vec s)]$.
\end{lemma}

%
\begin{proposition}\label{th:r_geq_2_bound}For $r\geq 2$, and any affine NRG $\GG$, $\CPoA(\GG,r)\leq \frac{r}{2}$. 
\end{proposition}
\begin{proof}
Let $\vec s^*$ be a global optimum of $\phi^r$, $\vec s_O=OPT(\GG)$. By Lemma~\ref{lemma:affine_optimal} $SC(\GG,\vec s^*) \leq \phi^r(\vec s^*) \leq \phi^r(\vec s_O) \leq \frac{r}{2} SC(\GG,\vec s_O)$.
%
\end{proof}


We can similarly derive a bound of $2/r$ for the range $r\in [1,2)$, but we can do better. 
We next show that as we increase the uncertainty level $r$ from $1$ towards $2$, we get a smooth improvement in social cost. 
 \begin{theorem}
\label{th:r_leq_2_bound}
For $r\in [1,2]$, and for any affine NRG $\GG$, $\CPoA(\GG,r) \leq 2-\frac{2}{r} + (\frac{2}{r}-1)\PoA(\GG) 
\leq \frac{2+2r}{3r}$.
\end{theorem}

Let $q^*(r) = \max_{\GG}\CPoA(\GG,r)$; and  $q^*_P(r) = \max_{\GG_P}\CPoA(\GG_P,r)$. 
The results above give us an upper bound on $q^*(r)$. By a careful analysis of the Pigou-type instances, we can compute $q^*_P(r)$ exactly, which also provides us with a lower bound on $q^*(r)$.  

%

\begin{corollary}
 For $r\in [1,2]$, $q^*(r) \in [\frac{4}{4r-r^2},\frac{2+2r}{3r}]$. For $r>2$, we have $q^*(r) \in [\frac{r^2}{4(r-1)},r/2]= [r/4+o(1),r/2]$. Also, $q^*_P(r)$ is equal to our lower bound on $q^*(r)$.\footnote{Roughgarden~\shortcite{roughgarden2003price} showed that for any class of cost functions, a worst-case example for the PoA can be constructed on an RSG with two resources. It is not clear if this is also true for the $\CPoA$, as in our case the optimum and the equilibrium are computed for two different games. However if a similar result can be proved, then our upper bounds would collapse to the lower bounds.}  
\end{corollary}
Note that for $r=1$ and $r=2$, we have $q^*(r)=q^*_P(r)$, and both are equal to the familiar values of $4/3$ and $1$, respectively. See Figure~\ref{fig:WCC_bounds} for a graphical comparison of the bounds.
\begin{figure}
%
%
\begin{center}

\begin{tikzpicture}[scale=0.9]

\begin{axis}[xlabel=$r$,ylabel=PoA,width=8.8cm,height=4.2cm, enlarge x limits = -1]
    \addplot[domain=1:2, red,    thin] {4/(4*x - x^2)};
		  \addplot[domain=2:9, red,  thin] {x^2/(4*(x-1))};
			\addplot[domain=1:2, thin, red] {(2+2*x)/(3*(x))};
			\addplot[domain=2:5.5, thin, red] {x/2};
		\addplot[domain=1:9, thin,dotted] {4/3};
		\addplot[domain=1:9, thin,dotted] {1};
		\addplot[domain=1:3.732, blue, dashed]{16/(8*(x+1/x) -(x+1/x)^2)};
			\addplot[domain=3.732:9, blue, dashed]{((x+1/x)^2)/(8*(x+1/x)-16)};
\end{axis}

\end{tikzpicture}
\caption{\label{fig:WCC_bounds}The solid red lines are upper- and
lower-bounds on
$q^*(r)$, i.e. the maximum PoA, across all games, 
for WCC players with uncertainty parameter $r$
(the lower bound is also exactly $q_P^*(r)$, the maximum PoA across
Pigou examples).
%
The blue dashed line is exactly  $q_P^{**}(r)$.
%
The dotted lines mark $4/3$ 
and $1$.}
\end{center}
\end{figure}
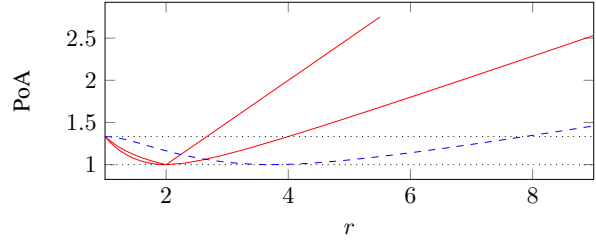

\newpar{Diverse population.}


We would like to show that if we have a population of agents with different uncertainty levels, the social cost does not exceed that of the  upper bound we
have on the worst type;  i.e., that $q^*(\vec r) \leq
q^*(r_i)$ for some type $i$ in the mixture.
We show something very close (our bound is slightly worse when there are types both below and above $r=2$).

Let $j=\argmin_i r_i, k = \argmax_{i}r_i$; let
$\alpha_j=\CPoA(\GG,r_j)$ if $r_j<2$ and $1$ otherwise; let
$\alpha_k=\frac{r_k}{2}$ if $r_k>2$ and $1$ otherwise.
%
\begin{lemma}
\label{lemma:ordered_r}
Let $\vec s,\vec s'$ be any two states in
affine NRG $\GG$,
and consider $r_3 \geq r_2 \geq r_1$. If $\phi^{r_3}(\vec s) \geq \phi^{r_3}(\vec s')$ and $\phi^{r_2}(\vec s) \leq \phi^{r_2}(\vec s')$, then $\phi^{r_1}(\vec s) \leq \phi^{r_1}(\vec s')$ as well. 
\end{lemma}

\begin{theorem}
\label{th:diverse_max}
Let $\GG$ be a symmetric affine RSG, $\vec r$ be an uncertainty vector. Then $\CPoA(\GG,\vec r) \leq \alpha_j \cdot \alpha_k$.
\end{theorem}


\begin{proof} [Proof of Theorem~\ref{th:diverse_max} for $r_j\geq 2$]
Let $\vec s^*$ be an equilibrium of $\GG^*(\vec r)$. 
 Let $\vec s^*_i$ be a state minimizing $\phi^{r_i}(\vec s)$. Note that $\phi^r(\vec s)<\phi^{r'}(\vec s)$ for all $r<r'$ and any $\vec s$.
 By Lemma~\ref{lemma:affine_optimal}, $\phi^2(\vec s)=SC(\GG,\vec s)$ for any $\vec s$. 
We next bound $SC(\GG,\vec s^*)$, dividing into cases: $r_j\geq 2$; $r_k\leq 2$; $r_j < 2 <r_k$.  
We prove for $r_j\geq 2$. 

Consider the state $\vec s^*_k$, which is an equilibrium of the game $\GG^*(r_k)$. If $\phi^{r_k}(\vec s^*)=\phi^{r_k}(\vec s^*_k)$, then $\vec s^*$ is also an equilibrium of $\GG^*(r_k)$. Since $\GG^*(r_k)$ is an NRG, all equilibria have the same social cost, thus 
$$\phi^{2r_k}(\vec s^*) = SC(\GG^*(r_k),\vec s^*) = SC(\GG^*(r_k),\vec s^*_k) = \phi^{2r_k}(\vec s^*_k).$$
As $2r_k > r_k \geq 2$, by Lemma~\ref{lemma:ordered_r} $\phi^{2}(\vec s^*_k) \geq \phi^{2}(\vec s^*)$ (in fact equal).  Thus

\labeq{SC_s_sk}
{~~~~~~~~~~~~~~~SC(\vec s^*) = \phi^2(\vec s^*) \leq \phi^2(\vec s^*_k) = SC(\vec s^*_k).}

Thus suppose that $\phi^{r_k}(\vec s^*_k),\phi^{r_k}(\vec s^*)$ differ. By  definition, $\phi^{r_k}(\vec s^*_k) < \phi^{r_k}(\vec s^*)$.
There must be some (non zero measure of) agents with different actions in both states. It cannot be that all of these agents have uncertainty $r_k$, since the states $\vec s^*,\vec s^*_k$ have a different $\phi^k$ potential.

Consider any such agent of type $i$, $r_i<r_k$, whose action under $\vec s^*$ differs from the one in $\vec s^*_k$.  If $\phi^{r_i}(\vec s^*_k) < \phi^{r_i}(\vec s^*)$ then $\vec s^*$ is not an equilibrium of $\GG^*(\vec r)$, as some agents of type $i$ would deviate.  We conclude that there is at least one type $i$ s.t. $r_i<r_k$, and $\phi^{r_i}(\vec s^*_k) \geq \phi^{r_i}(\vec s^*)$. 

Since $r_k> r_i \geq 2$, and $\phi^{2}(\vec s^*_k) \geq \phi^{2}(\vec s^*)$ by Lemma~\ref{lemma:ordered_r}, we get Eq.~\eqref{eq:SC_s_sk} again. Thus $SC(\vec s^*) \leq SC(\vec s^*_k)$.

By Lemma~\ref{lemma:affine_optimal}, we have that $\CPoA(\GG,\vec r) \leq \CPoA(\GG,r_k) \leq \alpha_k$ (note that $\alpha_j=1$).
\end{proof}

%
Finally, in a RSG a small fraction of the agents with high uncertainty cannot inflict too much damage.
\begin{theorem}\label{th:diverse_small}
Let $\GG$ be an affine RSG. Suppose that $\vec r$ is composed of two types, $r_k>r_j$. Then $\CPoA(\GG,\vec r)\leq \frac{r_j}{2} + \frac{n_k}{n}O(r_k)$.
\end{theorem}

%
%
%

\section{Worst-Case Regret}


For what follows, we assume that $r_i=r$ for all $i$. In addition,
we focus on RSGs, as the analysis is non-trivial even for such simple games.

\newpar{Equilibrium and convergence.}

%

In a RSG, the set of edges
$E$ is also the set of actions.  For every resource $e$ in state $\vec s$, we have 
%

$$c^{**}(\{e\},\os)= WCR(e,\os) =  c_e(r \ol s_e) - \min_{d\neq e}c_d(\ol s_d/r).$$

As $\GG$ is an RSG, $\vec s\in EQ(\GG^{**}(r))$ if and only if   $WCR(e,\os)$ is the same for all occupied edges, and at least as high in unoccupied edges.
Denote 
 $MR(\vec s) = \max_{e\in E:\ol s_e>0}WCR(e,\os)$.

%
%
%
Recall that every RSG with player-specific costs has a weak potential~\cite{milchtaich1996congestion}. We show a similar result for an RSG played by WCR players. Our result requires an additional technical property, but 
allows an explicit construction of the potential.
We say that a function $z(t)$ is {\em $r$-convex} if $z'(t/r)/r \leq r\cdot z'(r\cdot t)$ for all $t$, where $c'$ is the derivative of $c$. We note that $r$-convexity holds for convex and other commonly used functions (see Appendix~\ref{apx:regret}).
\begin{proposition}\label{th:WCR_potential} Consider an RSG $\GG$, and suppose that all cost functions are $r$-convex. Then $MR(\vec s)$ is a weak potential function of $\GG^{**}(r)$. 
Then if there are WCR moves, there is a WCR move that reduces $MR(\vec s)$.
\end{proposition}

\dcp{right?} \rmr{we already know en equilibrium exists in any nonatomic game}

\newpar{Equilibrium quality for Pigou instances.}

In the special case of the family of Pigou instances, we have:
%

\begin{proposition}
\label{th:upper_WCR_ge2} Let $q^{**}_P(r) = \max_{\GG_P} \RPoA(\GG_P,r)$.
\begin{enumerate}
	\item[(1)] For any $r \!\in\! [1,2\!+\!\sqrt{3}]$, we have $q^{**}_P(\!r\!)= \frac{16}{8(r\!+\!\frac{1}{r}) - (1\!+\!\frac{1}{r})^2}$.
\item[(2)] For any $r \geq 2+\sqrt{3}$, we have $q^{**}_P(r) =  \frac{(r+\frac{1}{r})^2}{8(r+\frac{1}{r})-16)}$.
\end{enumerate}
\end{proposition}

Thus $\RPoA(\GG_P,r)$ for $r\geq 2+\sqrt{3}$ is an increasing function in $r$,
and asymptoting to $r/8 + o(r)$ (see Figure~\ref{fig:WCC_bounds}). 
It remains an open question to derive an upper bound for games over more complex networks.

\section{Discussion}

%

\paragraph{Related work.}
We focus on previous work on uncertainty (especially strict uncertainty) 
in congestion games (CG).
%
Strict uncertainty  
been considered by Ashlagi et al.~\shortcite{ashlagi2006resource,ashlagi2009two}. They analyze ``safety-level'' strategies (similar to WCC) for agents who do not know the total number of players $k$, but only an upper bound 
on $k$, and focus on proving the existence of a symmetric mixed equilibrium. For atomic RSGs, they show that (as in our case) uncertainty improves the social welfare. 
Both the the analysis techniques and the reasoning required from agents in  their setting are
quite complex, despite the focus on a simple class of games, whereas in our case the game with WCC behavior reduces to a modified (player-specific) congestion game.

The next two papers are closer to our approach, where agents react to some noisy variation of the current state.
Meir et al.~\shortcite{MTBK12} study the effect of Bayesian uncertainty 
due to agents who may
fail to play with some probability (in the spirit of trembling-hand perfection). They focus mainly on RSGs and show that the PoA generally improves if failure probabilities are negligible, but not if they are bounded from $0$. 
	 Angelidakis et al.~\shortcite{angelidakis2013stochastic}
study a related model of RSGs with 
agents who react based on a  quantile of the cost distribution. We further discuss these two papers below.



Piliouras et al.~\shortcite{piliouras2013risk} also study CGs with
strict uncertainty with motivation similar to ours, and look at a wide range of decision-making
approaches including WCC and WCR. However,  agents in their model know
the actions of others exactly, and uncertainty over costs stems from  the
unknown \emph{order} in which players arrive. Our model is more direct, and
closer to the traditional view of congestion games.
  %
	%
%
%

Babaioff et al.~\shortcite{babaioff2007congestion}  study the effect of
introducing a small fraction $v$ of malicious agents.
This is related to our WCC model, where agents behave \emph{as
  if} the load on each edge is increased by $v=(r-1)\ol s_e$ additional
(malicious) agents. However, in
our model this added load only affects the behavior of  real agents, and does not affect the outcome directly.
Babaioff et al.~\shortcite{babaioff2007congestion}
 observe that adding some amount
of malicious agents may decrease the social
cost in equilibrium; see also followup work \cite{blum2008regret,roth2008price}.

Lastly, Halpern and Pass~\shortcite{halpern2012iterated} suggested \emph{iterated regret minimization} as an explanation to human behavior in many games. We emphasize that in contrast to our model, such behavior requires agents to explicitly reason about the \emph{incentives} of  other agents.

\newpar{Distance-based uncertainty.}

Our epistemic model adapts the multiplicative
distance-based uncertainty model, initially introduced in the context of voting theory~\cite{MLR14,Meir15}. There, the state was the number of votes
for each candidate, analogous to the measure of agents on each
resource in the present setting. 
While the epistemic model in both papers is derived from a similar approach, both the behavioral heuristics and the techniques for analysis are quite different. For example in voting games payoffs are highly discontinuous, so it is not {\em a priori} clear whether pure equilibria exist. 


There may appear to be a contradiction between our notion of uncertainty and equilibrium play: if agents converge to a particular state and play it repeatedly, then after a while we might expect them to be certain about this state. However even in ``standard'' congestion game the state (action profile) is only an abstraction of reality, where there is noise from various sources-- from players' actions and failures, to varying costs. Thus an equilibrium is a fixed point in the abstract model, although in reality there remain fluctuations around the equilibrium point. Thus even in equilibrium there may be some uncertainty about the exact loads. 

Some papers model the underlying distribution explicitly (e.g.,\cite{MTBK12,angelidakis2013stochastic}), and assume a belief structure that is derived from this distribution. 
Such an approach does not necessarily provide a better description of the way human players perceive the game. In our model we avoid such an explicit description, and instead use $\vec s$ (as an abstraction of the current state) and derive agents' beliefs directly from this state using the distance metric. These beliefs may or may not be consistent with the ``real'' underlying distribution, which may be highly complex. This simple belief structure  allows us to derive PoA bounds on a much wider class of games. 


Distance-based uncertainty can also be derived from a statistical viewpoint.
Suppose that an agent believes that the actual load is distributed around the reference load $\ol s_e$. A simple heuristic  considers a confidence interval around $\ol s_e$, with the size of the interval
modeled through $r_i$.   Under this interpretation, $r_i$ is higher for agent
types that are either more risk-averse or less-informed.\footnote{\label{fn:noise}For example, if the noise on edge $e$ is a normal distribution with standard deviation $\sigma$ and mean $0$, and an agent of type $i$ requires confidence of $95\%$ (roughly two standard deviations), then this translates to a strict uncertainty interval of $[\ol s_e-2\sigma,\ol s_e+2\sigma]$ 
} 
A crucial point is that if agents act independently
 the actual congestion would be highly concentrated around its expectation (that is, $r_i$ should approach $1$ as the size of the population grows), and for nonatomic agents there should be no uncertainty at all. However, experimental work in behavioral decision making suggests people perceive uncertainty over quantities as if the standard deviation is proportional to the expectation, even when this is false~\cite{kahneman1974subjective,tversky1974judgment}.\footnote{\label{fn:KT}The most famous example is an experiment where subjects are told the average number of girls born daily in a hospital is $s$. People believe that the probability that on a given day the number is within $[(1-r')s,(1+r')s]$ is fixed and does not depend on $s$. Kahneman and Tversky~\shortcite{kahneman1974subjective} highlight the contrast with standard, 
statistical analysis, where the range $r'$ is proportional to $1/\sqrt s$.} 

While in our model the uncertainty is over the actions of the other agents, an alternative way is to present it as strict uncertainty over the agent's own cost function, also known as Knightian Uncertainty~\cite{Knight21}. Chiesa et al.~\shortcite{chiesa2014knightian} have recently applied Knightian uncertainty to auctions, where an agent may have a valuation for an item, but only be aware of some \emph{interval} in which this valuation resides. Further studying the conceptual and technical connection between distance-based uncertainty and Knightian uncertainty may help to gain better understanding of both concepts.

	\subsection{Conclusions} 

	%
Game-theoretic models
        should explain and predict the behavior of
        players in games. 
%
Merely adding to such models uncertainty about the environment is insufficient, since as Simon~\shortcite{simon57} wrote: ``{\em ...the state of information may as well be regarded as a characteristic of the decision-maker as a characteristic of his environment}.''
Indeed, psychological studies suggest that people
        are both risk-averse and avoid probabilistic calculations~\cite{tversky1974judgment,slovic1980facts},
and raise concerns with standard
models of rationality.  
We
        believe that our model  captures these behavioral
        assumptions, and that it is simpler than other approaches for
        uncertainty representation.
In addition, we show that the model
        still permits the use of
standard game-theoretic tools such as
        equilibrium analysis. 
The model is flexible, and variations of the belief structure (the distance metric) can be easily made. One limitation of our approach is that distance based uncertainty (an interval) cannot capture bi-modal scenarios. For example, when congestion is usually mild, but in rare cases (an accident) congestion is very high.  This is similar to the reason that a normal distribution cannot capture such a scenario. 

Our results show that in a risk-averse population, an intermediate level of uncertainty helps to align the incentives of the agents with those of the society. This message is emphasized by showing similar results for two different interpretation of risk-aversion, and is consistent with findings from other models of uncertainty (see Related Work). 

Finally, lab experiments with routing games show that human subjects converge to states that are close to, but do not coincide with the  Nash equilibria~\cite{rapoport2009choice} (especially when looking at individual behavior rather than aggregate congestion).  
We believe that these discrepancies might be at least partially due to bounded-rational behavior of the agents, who may be risk-averse and/or applying simple variations of best-response heuristics. Thus our model might be able to better explain observed outcomes in congestion games.  More empirical and experimental work is required in that respect.
%
%
\dcp{this comment about explaining discrepancies made me thinking about
halpern and pass work on iterated regret minimization. might be worth a cite
somewhere, e.g. in a footnote. good defensive thing to do, i think.}



			
	


\bibliography{uncertainty}
\newpage
\title{Online appendices for paper \#21}
\maketitle

\onecolumn
\renewcommand{\thesection}{\Roman{section}}

\section{Proofs}

\begin{rproposition}{th:example_cycles}
There is an RSG $\GG$ with 3 resources, and a vector $\vec r$ s.t. $\GG^*(\vec r)$ contains a cycle of (infinitesimal) best responses.
\end{rproposition}
\begin{proof}
We first present the example constructed by Milchtaich~\shortcite{milchtaich1996congestion}.
There are three resources, $x,y,z$, and three atomic agents.
We denote by $\vec a=(e_1,e_2,e_3)\in \{x,y,z\}^3$ the actions of the three agents at a given state. 
Milchtaich constructs three cost functions for every resource such that 
\begin{align*}
c^*_{1,z}(1) < c^*_{1,y}(1); c^*_{1,y}(2) < c^*_{1,z}(2) \\
c^*_{2,z}(2) < c^*_{2,x}(2); c^*_{2,x}(1) < c^*_{2,z}(1) \\
c^*_{3,y}(1) < c^*_{3,x}(1); c^*_{3,x}(2) < c^*_{3,y}(2),
\end{align*}
and shows that this leads to a cycle of best-responses through the profiles $(y,x,x)\step{} (z,x,x)\step{} (z,z,x)\step{} (z,z,y)\step{} (y,z,y) \step{} (y,x,y)$ and back to $(y,x,x)$.
 
Our (nonatomic) RSG $\GG$ has 3 types of agents, each with a mass of one unit.
 Let $r_1=1,r_2=r_3=10$. Set cost functions on the three edges $E=\{x,y,z\}$ so that 
	$$c_z(1) < c_x(1)<c_y(1) <c_y(2)<c_z(2) < c_y(10)<c_x(10)<c_z(10)<c_z(20)<c_x(20)<c_y(20).$$
	Since the only constraint is that each $c_e$ is monotone, this is always possible.
	We observe that 
	\begin{align*}
c^*_{1,z}(1) = c_z(1) < c_y(1) = c^*_{1,y}(1); c^*_{1,y}(2)  = c_y(2) < c_z(2) = c^*_{1,z}(2) \\
c^*_{2,z}(2)  = c_z(20) <  c_x(20) =c^*_{2,x}(2); c^*_{2,x}(1) = c_x(10) < c_z(10)= c^*_{2,z}(1) \\
c^*_{3,y}(1) = c_y(10) < c_x(10) =c^*_{3,x}(1); c^*_{3,x}(2) = c_x(20) < c_y(20) = c^*_{3,y}(2).
\end{align*}
Denote by $(e_1,e_2,e_3)$ the state where \emph{all} type~$i$ agents use resource $e_i\in\{x,y,z\}$.
We get a similar cycle when we start from the state $\vec a^0 =(y,x,x)$. First, all type~1 agents move to $z$, since $c^*_{1,z}(t)<c^*_{1,y}(1-t)$ for all $t\leq 1$, so we get state $\vec a^1 = (z,x,x)$. Then all type~2 agents move from $x$ to $y$, and so on. 
 
We need not assume that all agents of the same type move as one, that is, we can get from $\vec a^0$ to $\vec a^1$ is a sequence of steps, as long as only type~1 agents move. 
	\end{proof}

	\begin{rlemma}{lemma:affine_optimal}
Let $\GG$ be an affine NRG, and suppose that $r\geq 2$. Then  for all $\vec s$,
	 $\phi^r(\vec s) \in [SC(\GG,\vec s),\frac{r}{2}SC(\GG,\vec s)]$.
\end{rlemma}

\begin{proof}  
\begin{align*}
\phi^r(\vec s) = \phi(\GG^*(r),\vec s)  &= \sum_e \int_{t=0}^{\ol s_e}c^*_e(t) dt =  \sum_e \int_{t=0}^{\ol s_e}c_e(r \cdot t) dt \\
&= \sum_e \int_{t=0}^{\ol s_e}(a_e \cdot r \cdot t + b_e ) dt 
 =\sum_e [\frac{r}{2} a_e \cdot t^2 + b_e t]_{t=0}^{\ol s_e}  \\
& =\sum_e \ol s_e (\frac{r}{2} a_e \cdot \ol s_e + b_e).
\end{align*}
Now, if $r\geq 2$, then  $\phi^r(\vec s) = \sum_e \ol s_e (\frac{r}{2} a_e \cdot \ol s_e + b_e) \geq \sum_e (\ol s_e a_e \cdot \ol s_e + b_e)=SC(\GG,\vec s)$; and 
$$\phi^r(\vec s) =\sum_e \ol s_e (\frac{r}{2} a_e \cdot \ol s_e + b_e) \leq \frac{r}{2}\sum_e \ol s_e ( a_e \cdot \ol s_e + b_e) = \frac{r}{2}SC(\GG,\vec s).$$
\end{proof}

\begin{rtheorem}{th:r_leq_2_bound}
For $r\in [1,2]$, and for any affine NRG $\GG$, $\CPoA(\GG,r) \leq 2-\frac{2}{r} + (\frac{2}{r}-1)\PoA(\GG)$. Taking the worst upper bound over all affine games, we get $\CPoA(\GG,r) \leq \frac{2+2r}{3r}$.
\end{rtheorem} 

\rmr{ The place where the analysis can be tightened is where we rely on convexity of $SC$ instead of writing it explicitly (note: It should be shown that in OPT at least some of the cost is due to the linear part $a_e$).} 

\begin{proof}
Let $\vec s'$ and $\vec s''$ be equilibrium points of $\GG$ and $\GG^*(2)$, respectively. 
Recall that $\vec s' = \argmin_{\vec s}\phi(\GG,\vec s)$, i.e., it  minimizes the potential function of $\GG$ over all real vectors $\vec s$, subject to some feasibility constraints. 

Taking the derivative of the potential function $\phi(\GG,\vec s)$ w.r.t. $\ol s_e$, we get 
$$g_e(\vec s)=\frac{\partial \phi(\GG,\vec s)}{\partial \ol s_e} = a_e \ol s_e + b_e.$$

Similarly, for a game $\GG^*(r)$, we get that 
$$h^{r}_e(\vec s)= \frac{\partial \phi^r(\GG,\vec s)}{\partial \ol s_e} = r \cdot a_e \ol s_e + b_e, $$
and in particular $h^2_e(\vec s) = 2a_e \ol s_e + b_e $.

We define $\vec s^* = \beta \vec s'' + (1-\beta)\vec s'$, where $\beta = \frac{2r-2}{r}$ (in particular $s^*_e = \beta \ol s''_e + (1-\beta)\ol s'_e$ for all $e$). Thus for $r=1,r=2$ we get $\vec s^*=\vec s'$ and $\vec s^*=\vec s''$, respectively. 
We claim that $\vec s^*$ is an equilibrium of $\GG^*(r)$. 
As for feasibility, since $\vec s^*$ is the convex combination of two valid states, and all feasibility constraints are linear, $\vec s^*$ is also a feasible state.

We next show that for $r\in[1,2]$, $h^r_e(s^*_e)=0$ for all $e$. That is, that $\vec s^*$ is the minimum of $\phi^r(\vec s)$ (and thus an equilibrium of $\GG^*(r)$. 
\begin{align*}
h^{r}_e(\vec s^*)&= r \cdot a_e s^*_e +  b_e \\
&= r \cdot a_e  (\beta \ol s''_e + (1-\beta) \ol s'_e)+  b_e  \\
&= r \cdot a_e  (\beta \ol s''_e + (1-\beta) \ol s'_e)+  (b_e)(r-1) + (b_e)(2-r)\\
&= \beta r \cdot a_e   \ol s''_e + (b_e + d_e)(r-1) + (1-\beta)  r \cdot a_e  \ol  s'_e +  (b_e)(2-r)\\
 &=  \frac{2r-2}{r} r \cdot a_e  \ol s''_e + (b_e)(r-1) + \frac{2-r}{r}  r \cdot a_e \ol  s'_e +  (b_e )(2-r)\\
 &=  2(r-1)  a_e \ol  s''_e + (b_e)(r-1) + (2-r) \cdot a_e  \ol s'_e +  (b_e)(2-r)\\
 &=  (r-1) (2 a_e  \ol s''_e + b_e) + (2-r)( a_e  \ol s'_e +  b_e)\\
& = (r-1)h^2_e(\ol s''_e) + (2-r)g_e(\ol s'_e) = 0+0 = 0.
\end{align*}
We next bound the social cost at $\vec s^*$:
\begin{align*}
SC(\GG,\vec s^*) &= SC(\GG,\beta \vec s'' + (1-\beta) \vec s') \\
&\leq \beta SC(\GG,\vec s'') +(1-\beta)SC(\GG,\vec s') \tag{convexity of $SC$}\\
&=\beta OPT(\GG) +(1-\beta)SC(\GG,\vec s') \tag{$\vec s''$ is optimal in $\GG$}\\
&=\frac{2r-2}{r}OPT(\GG) + \frac{2-r}{r}OPT(\GG)\PoA(\GG) \\
&= \left(\frac{2r-2}{r} + \frac{2-r}{r}\PoA(\GG)\right) OPT(\GG),
\end{align*}
thus $\CPoA(\GG,r) \leq  \frac{2r-2}{r} + \frac{2-r}{r}\PoA(\GG)$.

Finally, since $\PoA(\GG)\leq \frac43$ for any affine game, 
$$\CPoA(\GG,r) \leq \frac{2r-2}{r} + \frac{2-r}{r}\frac{4}{3} = \frac{3(2r-2)+(2-r)4}{3r} = \frac{2+2r}{3r},$$
which concludes the proof.
\end{proof}

\begin{proposition}
\label{th:upper_RSG_ge2}
For any Pigou instance $\GG_P(a_2,n)$, and any $r\geq 2$, we have $\CPoA(\GG,r) \leq  \frac{r^2}{4(r-1)}$, and this bound is tight.
\end{proposition}
\begin{proof}
We denote $a=a_2$, and $a'=ar$. We denote the optimal state by $\vec s'$ (which is the equilibrium for $r'=2$). For $r>2$ we always have $s^*_1\geq s'_1$ as the cost of $e_2$ has more effect on agents with higher uncertainty.

The game $\GG^*(r)$ has a unique equilibrium, where either one of the resources has no agents, or $a_2 r \ol s_2 = c_2(r \ol s_2) = c_1(r \ol s_1) = 1$. Also, it is easy to check that $\ol s^*_2>0$, and that if $\ol s^*_1=0$, then $\vec s^*=\vec a'$ and is thus optimal. Thus $\ol s^*_2 = \frac{1}{ar}, \ol s^*_1 = n-\frac{1}{ar}$  The social welfare in $\vec s^*$ can be written as 
$$SC(\vec s^*) = (n-\frac{1}{ar})1 + \frac{1}{ar} a \frac{1}{ar} = n-\frac{1}{ar}+\frac{1}{ar^2}.$$

Suppose first that $\ol s'_1=0$. Thus $\ol s'_2 = n$, and $SC(\vec s') = n^2a$. 
Thus 
\begin{align*}
\CPoA(\GG,r) &= \frac{SC(\vec s^*)}{SC(\vec s')} = \frac{n-\frac{1}{ar}+\frac{1}{ar^2}}{n^2a} = r\frac{n-\frac{1}{a'}+\frac{1}{a'r}}{n^2a'} \\
& = r(\frac{1}{na'} - \frac{1}{(na')^2} + \frac{1}{(na')^2 r})  \tag{denote $x=\frac{1}{na'}$}\\
& = r(x - x^2 + x^2/r) = r(x-x^2\beta), \tag{for $\beta = 1-\frac{1}{r}$}
\end{align*}
The last expression is maximized for $x = \frac{1}{2\beta}$, thus
\begin{align*}
\CPoA(\GG,r) &\leq r(\frac{1}{2\beta} - \frac{1}{4\beta}) = \frac{r}{4\beta} = \frac{r}{4(1-\frac{1}{r})}
\end{align*}

Next, suppose that $\ol s'_1>0$. Then $\ol s'_2=\frac{1}{2a}, \ol s'_1 = n-\frac{1}{2a}$, and 
$SC(\vec s') = n-\frac{1}{2a}+\frac{1}{4a}$. Note that this entails $n> \ol s'_2 = \frac{1}{2a}$.
In this case
\begin{align*}
\CPoA(\GG,r) &= \frac{SC(\vec s^*)}{SC(\vec s')} = \frac{n-\frac{1}{ar}+\frac{1}{ar^2}}{n-\frac{1}{2a}+\frac{1}{4a}} \\
& \leq \frac{ \frac{1}{2a}-\frac{1}{ar}+\frac{1}{ar^2}}{ \frac{1}{2a}-\frac{1}{2a}+\frac{1}{4a}} \tag{by our bound on $n$}\\
& = 4a(\frac{1}{2a}-\frac{1}{ar}+\frac{1}{ar^2}) = 2 - \frac{4}{r} + \frac{4}{r^2} < 2.
\end{align*}
Thus in the second case we only get a constant price of anarchy.

For tightness, it is sufficient to look at $\GG_P(a,n)$ for any $a,n$ s.t. $an=\frac{1}{r-1}$. Then the inequality we had in the first case becomes an equality.
\end{proof}

\begin{proposition}
\label{th:upper_RSG_le2}
For any Pigou instance $\GG_P(a_2,n)$, and any $r\in [1,2]$, we have $\CPoA(\GG,r) \leq  \frac{4}{4r-r^2}$, and this bound is tight.
\end{proposition}
\begin{proof}
We denote the equilibrium and the optimum by $\vec s^*,\vec s'$, respectively.
For $r\leq 2$ we always have $\ol s^*_1\leq \ol s'_1$, by a symmetric argument to the one above.
We consider three cases.

Case~I: If $\ol s'_1=0$ then $\ol s^*_1=0$ as well. Then $\vec s^*=\vec s'$ and $\CPoA(\GG,r)=1$. 

Case~II: Suppose that $\ol s^*_1>0$ and $\ol s'_1>0$. This is similar to the same case in the previous proof, except we use the fact that $n> \ol s^*_2 = \frac{1}{ar}$. Thus
\begin{align*}
\CPoA(\GG,r) &= \frac{SC(\vec s^*)}{SC(\vec s')} = \frac{n-\frac{1}{ar}+\frac{1}{ar^2}}{n-\frac{1}{4a}} \\
& < \frac{ \frac{1}{ar}-\frac{1}{ar}+\frac{1}{ar^2}}{ \frac{1}{ar}-\frac{1}{4a}} \tag{by our bound on $n$}\\
& = \frac{ 1}{ar^2( \frac{1}{ar}-\frac{1}{4a})} = \frac{ 1}{r-\frac{r^2}{4}} = \frac{4}{4r-r^2}.
\end{align*}

Case~III: Suppose that $\ol s^*_1=0,\ol s'_1>0$. Then $SC(\vec s^*) = n^2a$ and $SC(\vec s') = n-\frac{1}{4a}$ as above. Thus $n = \ol s^*_2 \leq \frac{1}{ar}$. Also
\begin{align*}
\CPoA(\GG,r) & = \frac{n^2a}{n-\frac{1}{4a}} =  4a\frac{n^2a}{4an-1} = 4\frac{(an)^2}{4an-1} \\
& = 4\frac{x^2}{4x-1}  \tag{for $x=an$}\\
&\leq 4\frac{\frac{1}{r^2}}{4/r-1} = \frac{4}{4r-r^2}
\end{align*}
To see why the inequality holds, note the following:
Note that $x=an \geq \frac{1}{2}$, and the function is increasing in $x$ in this range (attains a minimum at $x=\frac12$). Thus we can replace $x$ by its upper bound $\frac{1}{r}$ 

For tightness we can take any $\GG_P(a,n)$, as long as $an = \frac{1}{r}$. Then the inequality in Case~III becomes an equality.
\end{proof}

\begin{rlemma}{lemma:ordered_r}
Let $\vec s,\vec s'$ be any two states in
affine NRG $\GG$ with three uncertainty types, 
and consider $r_3 \geq r_2 \geq r_1$. If $\phi^{r_3}(\vec s) \geq \phi^{r_3}(\vec s')$ and $\phi^{r_2}(\vec s) \leq \phi^{r_2}(\vec s')$, then $\phi^{r_1}(\vec s) \leq \phi^{r_1}(\vec s')$ as well. 
\end{rlemma}

\begin{proof}
Intuitively, we show that every from one state to another induces a cutoff point $r^*$ over types, such that either all agents with $r_i>r^*$ gain and the others lose, or vice versa. 

Given the two states $\vec s,\vec s'$, define the function $z(r) = \phi^r(\vec s) - \phi^r(\vec s')$. Observe that 
$$z(r)=\phi^{r}(\vec s) - \phi^{r}(\vec s') = \frac{r}{2}\sum_e a_e ((\ol s_e)^2-(\ol s'_e)^2) + \sum_e b_e (\ol s_e-\ol s'_e) = r Z_1 + Z_2$$
for some constants $Z_1,Z_2$. Thus $z(r)$ is monotone in $r$ (either non-increasing or non-decreasing). 

By the premise of the lemma, $z(r_3)> 0,z(r_2)\leq 0$, thus $z(r)$ is strictly decreasing. We conclude that $z(r_1)\leq 0$, which completes the proof. 
\end{proof}

\begin{rtheorem}{th:diverse_max}
Let $\GG$ be a symmetric affine RSG, $\vec r$ be an uncertainty vector. Then $\CPoA(\GG,\vec r) \leq \alpha_j \cdot \alpha_k$.
\end{rtheorem}
\begin{proof}[Completion of the proof]
Case~II: $r_k\leq 2$. By a symmetric proof, we get that $\CPoA(\GG,\vec r) \leq \CPoA(\GG,r_j)= \alpha_j$. 

Case~III:  $r_j < 2 <r_k$. We repeat a similar argument to Case~II, only using $r_k$ instead of $2$, to get that $\phi^{r_k}(\vec s^*) \leq \phi^{r_k}(\vec s^*_j)$. Thus by Lemma~\ref{lemma:affine_optimal},

$$SC(\vec s^*) \leq \phi^{r_k}(\vec s^*) \leq \phi^{r_k}(\vec s^*_j) \leq \frac{r_k}{2} SC(s^*_j) \leq  \frac{r_k}{2}\alpha_j SC(OPT),$$

which completes the proof. Also note that $\alpha_j\leq 2$ so in any case $\vec s^*$ is at most as bad as $r_k\cdot SC(OPT)$. 
\end{proof}

\begin{rtheorem}{th:diverse_small}
Let $\GG$ be an affine RSG. Suppose that $\vec r$ is composed of two types, $r_k>r_j$. Then $\CPoA(\GG,\vec r)\leq \frac{r_j}{2} + \frac{n_k}{n}O(r_k)$.
\end{rtheorem}
\begin{proof}
Let $\vec s^*$ be an equilibrium of $\GG^*(\vec r)$. Note that in the load state $\os^*$, every edge load is composed of two parts $\ol s^*_e = s^*_{j,e} + s^*_{k,e}$. 
Let $E_i = \{e\in E : s^*_{i,e}>0\}$.
Denote by $N_j,N_k$ the actual sets of all type~$j$ and type~$k$ agents.
 We can sum the costs of the two agent types independently, so that $SC(\GG,\vec s^*) = \sum_{e\in E_k}  s^*_{k,e} c_e(\ol s^*_e) + \sum_{e\in E_j}  s^*_{j,e} c_e(\ol s^*_e)$. Let $OPT$ be an optimal state for $\GG$.

Suppose we turn all the type~$k$ agents to type~$j$ agents, and check if there are any moves from the state $\vec s^*$. Note that there can be no moves from $E_j$ to other edges (since then there are type~$j$ agents in $\vec s^*$ with an improvement move). Thus the ``new'' type~$j$ agent can only hurt the existing ones. Formally,  if we continue until convergence to $\vec s^{*j}$ (the equilibrium of $\GG^*(r_j)$) then all agents in $N_j$ have a cost in $\vec s^{*j}$ that is at least as high as in $\vec s^*$:
$$\sum_{e\in E_j} s^*_{j,e} c_e(\ol s^*_e) \leq \sum_{e\in E_j} \ol s^{*j}_e c_e(\ol s^{*j}_e) \leq SC(\GG,\vec s^{*j}) \leq \frac{r_j}{2}SC(OPT).$$

Similarly, 
$$\sum_{e\in E_k}  s^*_{k,e} c_e(\ol s^*_e) \leq \sum_{e\in E_k} s^*_{k,e} c_e(\ol s^{*k}_e).$$
In $\vec s^{*k}$, all agents have the same experienced cost $c_f(r_k \vec s^{*k}) = F$. Thus $SC(\GG,\vec s^{*k}) \leq SC(\GG^*(r_k),\vec s^{*k}) = nF$.

In particular, the total experienced cost of $N_k$ in $\vec s^{*k}$ is $n_k F$ (and this is higher than the actual cost). We have that
$$ nF  = SC(\GG^*(r_k),\vec s^{*k}) \leq \frac43 SC(\GG^*(r_k),OPT) \leq \frac{4 r_k}{3} SC(\GG,OPT),$$
thus
\begin{align*}
\sum_{e\in E_k} s^*_{k,e} c_e(\ol s^*_e)&\leq \sum_{e\in E_k} s^*_{k,e} c_e(\ol s^{*k}_e) \leq n_k F  = \frac{n_k}{n} nF\leq   \frac{n_k}{n} \frac{4 r_k}{3} SC(OPT) &\Rightarrow\\
SC(\vec s^*) &\leq \frac{r_j}{2}SC(OPT) + \frac{n_k}{n}  \frac{4 r_k}{3} SC(OPT) = ( \frac{r_j}{2} + \frac{n_k}{n}  O(r_k)) SC(OPT).
\end{align*}
This shows that $\CPoA(\vec GG,\vec r)\leq  \frac{r_j}{2} + \frac{n_k}{n}  O(r_k)$.
\end{proof}

\section{Regret minimization}
\label{apx:regret}

\begin{lemma}
Any convex function $z(t)$ is $r$-convex for any $r\geq 1$. 
\end{lemma}
\begin{proof}
For any convex function, $z'(t)$ is a non-decreasing function. Thus
$z'(t/r)/r \leq z'(t/r) \leq z'(tr) \leq r z'(tr).$
\end{proof}
%

\begin{lemma}
Any polynomial function $z(t)$ is $r$-convex for any $r\geq 1$. 
\end{lemma}
\begin{proof}
We write $z(t) = \sum_{j=1}^J a_j t^{b_j}$, where $b_j\geq 0$ for all $j$. Then $z'(t) = \sum_{j=1}^J  j a_j t^{b_j-1}$.
\begin{align*}
z'(t/r)/r &= \frac{1}{r}\sum_{j=1}^J  j a_j (t/r)^{b_j-1} = r^{-1}\sum_{j=1}^J  j a_j t^{b_j-1} r^{1-b_j}\\
&= \sum_{j=1}^J  j a_j t^{b_j-1} r^{-b_j} \leq \sum_{j=1}^J  j a_j t^{b_j-1} r^{b_j}  = r\sum_{j=1}^J  j a_j t^{b_j-1} r^{b_j-1} = z'(rt)r 
\end{align*}
\end{proof}
%

\begin{rproposition}{th:WCR_potential} Consider an RSG $\GG$, and suppose that all cost functions are $r$-convex. Then $MR(\vec s)$ is a weak potential function of $\GG^{**}(r)$. 
\end{rproposition}
That is with  WCR moves the game is \emph{weakly acyclic}-  there is always some WCR move that reduces $MR(\vec s)$.

\begin{proof}
Let $y\in \argmin_{e\in M}WCR(e,\os)$, i.e. a  resource where the WCR is minimum. Let $X=\{e\in M: WCR(e,\os)=MR(\vec s), \ol s_e>0\}$, i.e.,  the set of resources in the support of $\os$ of which WCR is maximum. 

If $y\in X$ then we are done, as no agent can improve her WCR
utility. Otherwise, we have that $MR(\vec s) > WCR(y,\os)$. We will
show that there is a rational move to a state $\vec s^+$ (that is, the
utility of all involved agents strictly improves in $\vec s^+$), that
reduces $MR(\vec s^+)<MR(\vec s)$.

Let $w$ be the resource with the lowest minimal cost $c_w( \ol s_w/r)$, and let $w'$ be the resource with the second lowest minimal cost (ordered arbitrarily if there is a tie).  Note that for any $e\neq w$, $WCR(e,\os) = c_e(r \ol s_e) - c_w(\ol s_w/r)$, and $WCR(w,\os) = c_w(r \ol s_w) - c_{w'}(\ol s_{w'}/r)$.

Intuitively, moving some agents from $e\in X$ to $y$ reduces $WCR(e,\vec s)$, but may increase $WCR(e',\vec s)$ for some other $e'\in X$, and thus increase $MR(\vec s)$. 
Hence we divide into three cases: (a) $w\notin X$; (b) $w\in X$ but $w'\notin X$; and (c) $w,w'\in X$.

In case (a), take a mass of $\eps$ of all $e\in X$ and move it to $y\notin X$. This decreases $c_e(r \ol s_e)$ for all $e\neq y$ (strictly decreases for all $e\in X$), and does not decrease $c_w(\ol s_w/r)$, since $\ol s_w$ either increases or remains unchanged. Thus $WCR(e,\vec s^+)\leq WCR(e,\vec s)$ for all $e\in M\setminus\{y\}$, with a strict inequality for $e\in X$. The only resource where WCR possibly increases is $y$.  Since $WCR(y,\vec s)<MR(\vec s)$, then by continuity there is a sufficiently small $\eps$ s.t. $WCR(y,\vec s^+) \leq MR(\vec s^+) < MR(\vec s)$. 

In case (b), take a mass of $\eps$ from every $e\in X\setminus \{w\}$, and move it to $w'$. As in case (a), the WCR of all resources $e\in X\setminus \{w\}$ strictly decreases. As for $w$, we have 
\begin{align*}
WCR(w,\vec s^+)& = c_w(r \ol s^+_w)-c_{w'}(\ol s^+_{w'}/r) = c_w(r \ol s_w)-c_{w'}((\ol s_{w'}+\eps(|X|-1))/r) \\
&< c_w(r \ol s_w)-c_{w'}(\ol s_{w'}/r) = WCR(w,\vec s).
\end{align*}

So once again for a sufficiently small $\eps$, $MR(\vec s^+)<MR(\vec s)$. 

Case (c)  is the most complicated case. We denote the derivative of the cost function $c_e(t)$ at point $t$ by $c'_e(t)$.  Since $t$ may itself be a function of $\ol s_e$, we define $\hat c_e(t(\ol s_e)) = \frac{\partial c_e(t( \ol s_e))}{\partial \ol s_e}$. 
Note that for a constant $\alpha$, $\hat c_e(\alpha \ol s_e) = \alpha c'_e(\alpha \ol s_e)$.

Since cost functions have a bounded derivative in the relevant range, then for some small $\eps$ we have 
$$c_e(\alpha (\ol s_e+\eps)) \cong  c_e(\alpha \ol s_e) + \eps \hat c_e(\alpha \ol s_e) =  c_e(\alpha \ol s_e) + \eps \alpha  c'_e(\alpha \ol s_e).$$

\rmr{can be more precise by placing lower and upper bounds on $c_e(\alpha (\ol s_e+\eps)) $ with $c'_e(\alpha \ol s_e)$ and $c'_e(\alpha (\ol s_e+\eps))$ (use convexity)}

By convexity, $c'_e(\ol s_e/r)\leq c'_e(r \ol s_e)$, and since $r\geq 1$, we have 
$$\hat c_e(r \ol s_e)  = rc'_e(r \ol s_e)\geq \frac{1}{r}c'(\ol s_e /r) = \hat c_e(\ol s_e/r)$$
for all $e$ (with strict inequality for $r>1$), and in particular for $w$ and $w'$. 

Case (c.1): Suppose that $\hat c_{w'}(\ol s_{w'}/r) <  \hat c_{w}(r \ol s_{w})$ and $\hat c_{w}(\ol s_{w}/r)< \hat c_e(r \ol s_e)$ for all $e\in X\setminus \{w\}$, then we can still take $\eps$ of all resources in $X$ and move it to $y\notin X$ (as in case~(a)). We get that in the new state $\vec s^+$, for some sufficiently small $\eps$:
\begin{align*}
WCR(w,\vec s^+) &= c_w(r \ol s^+_w) - c_{w'}(\ol s^+_{w'}/r) \cong (c_w(r \ol s_w)-\eps r c'_w(r \ol s_w))- (c_{w'}(\ol s_{w'}/r)-\eps\frac{1}{r} c'_{w'}(\ol s_{w'}/r) ) \\
& = WCR(w,\vec s) +  \eps (\hat c_{w'}(\ol s_{w'}/r) -  \hat c_w(r \ol s_w)) < WCR(w,\vec s).
\end{align*}
For any $e\in X\setminus\{w\}$,
\begin{align*}
WCR(e,\vec s^+) &= c_e(r \ol s^+_e) - c_{w}(\ol s^+_{w}/r) \cong (c_e(r \ol s_e)-\eps r c'_e(r \ol s_e))- (c_{w}(\ol s_{w}/r)-\eps\frac{1}{r} c'_{w}(\ol s_{w}/r) ) \\
& = WCR(e,\vec s) +  \eps (\hat c_{w}(\ol s_{w}/r) -  \hat c_e(r \ol s_e)) < WCR(e,\vec s).
\end{align*}
So as in case (a) we have $MR(\vec s^+)<MR(\vec s)$.

Case (c.2): Suppose that 
 some of the inequalities for $X$ are violated. For Every $e\in X$ s.t.  $\hat c_{w}(\ol s_{w}/r)\geq \hat c_e(r \ol s_e)$ define $\delta_e$ strictly between $\frac{\hat c_{w}(\ol s_{w}/r)}{\hat c_e(r \ol s_e)}$ 
and $\frac{\hat c_{w}(\ol s_{w}r)}{\hat c_e( \ol s_e/r)}$ 
 For $e\in X$ where the inequality was not violated, set $\delta_e=1$. 

 We move a mass of $\eps$ from $w$ to $y$, and a mass of $\delta_e\eps>\eps$ from all $e\in X\setminus\{w\}$  (including $w'$) also to $y$. Thus we get: 
\begin{align*}
WCR(w,\os^+) &= c_w(r \ol s^+_w) - c_{w'}(\ol s^+_{w'}/r) \cong (c_w(r \ol s_w)-\eps t \hat c_w(r \ol s_w))- (c_{w'}(\ol s_{w'}/r)-\delta_{w'}\eps \frac{1}{r}\hat c_{w'}(\ol s_{w'}/r) ) \\
& = WCR(w,\os) +  \eps (\frac{1}{r}\delta_{w'} \hat c_{w'}(\ol s_{w'}/r) -  r \hat c_w( \ol s_w r)) \\
& < WCR(w,\os) +  \eps (\frac{1}{r}\hat c_{w}(\ol s_{w} r) - r \hat c_w( \ol s_w r) ) < WCR(w,\os),
\end{align*}

whereas for all $e\in X\setminus\{w\}$ for which $\delta_e >1$ (possibly including $w'$), 
\begin{align*}
WCR(e,\os^+) &= c_{e}(r \ol s^+_{e}) - c_{w}(\ol s^+_{w}/r) \cong (c_{e}(r \ol s_{e})-\delta_e\eps r \hat c_{e}(r \ol s_{e}))- (c_{w}(\ol s_{w}/r)-\eps \frac{1}{r}\hat c_{w}(\ol s_{w}/r) ) \\
& = WCR(e,\os) +  \eps ( \frac{1}{r} \hat c_{w}(\ol s_{w}/r) -  r\delta_e \hat c_{e}(r \ol s_{e})) \\
& = WCR(e,\os) +  \eps ( \frac{1}{r} \hat c_{w}(\ol s_{w}/r) -  r\delta_e \hat c_{e}(r \ol s_{e})) \\
& \leq  WCR(e,\os) +  \eps ( \frac{1}{r}\hat c_{w}(\ol s_{w}/r) -  r \hat c_{w}(r \ol s_{w})) < WCR(e,\os).
\end{align*}
For $e\in X$ where $\delta_e=1$,
\begin{align*}
WCR(e,\os^+) &= c_{e}(r \ol s^+_{e}) - c_{w}(\ol s^+_{w}/r) \cong (c_{e}(r \ol s_{e})-\delta_e\eps r \hat c_{e}(r \ol s_{e}))- (c_{w}(\ol s_{w}/r)-\eps\frac{1}{r} \hat c_{w}(\ol s_{w}/r) ) \\
& = WCR(e,\os) +  \eps (\frac{1}{r} \hat c_{w}(\ol s_{w}/r) -  r \hat c_{e}(r \ol s_{e}))< WCR(e,\os).
\end{align*}

%
%

Thus as in the previous cases, the WCR cost goes down for all $X$, which means $MR(\vec s^+)<MR(\vec s)$.
\end{proof}

\begin{rproposition}{th:upper_WCR_ge2}
Consider any Pigou instance $\GG_P(a_2,n)$. 
\begin{enumerate}
	\item For any $r \in [1,2+\sqrt{3}]$, we have $\RPoA(\GG_P,r) \leq \frac{16}{8(r+1/r) - (1+1/r)^2}$, and this bound is tight.
\item For any $r \geq 2+\sqrt{3}$, we have $\RPoA(\GG_P,r) \leq \frac{(r+1/r)^2}{8(r+1/r)-16)} = \frac{1}{8}r + o(r)$, and this bound is tight.
\end{enumerate}
\end{rproposition}
\begin{proof}
We denote $a=a_2$. We denote the optimal state by $\vec s'$.

The game $\GG^*(r)$ has a unique equilibrium, If all agents use $e_2$, then $SC(\vec s^*)=n^2a$. This occurs if and only if $ran-1 = c^{**}_2(0,n) \leq c^{**}_1(0,n) = an/r - 1$, i.e. iff $n \leq \frac{2}{a(r+1/r)}$. 

Otherwise, we have $c^{**}_2(\os^*) = c^{**}_1(\os^*)$.
\begin{align*}
&ra \ol s^*_2 - 1 =  1-a \ol s^*_2/r & \iff\\
&(r+1/r)a \ol s^*_2 = 2& \iff \\
& \ol s^*_2 = \frac{2}{a(r+1/r)},\ol s^*_1 = n-\frac{2}{a(r+1/r)}.
\end{align*}
Thus the social cost in equilibrium is 
$$SC(\vec s^*) = (n-\frac{2}{a(r+1/r)}) + (\frac{2}{a(r+1/r)})^2 a = n-\frac{2}{a(r+1/r)} + \frac{4}{a(r+1/r)^2}.$$

For the optimal outcome we already know that $\ol s'_2=\frac{1}{2a},SC(\vec s') = 1-\frac{1}{4a}$ if $n\geq \frac{1}{2a}$ and $\ol s'_2=n,SC(\vec s')=n^2 a$ otherwise. 

Note that the cutoff value of $r$ is when $\frac{1}{2a} = \frac{2}{a(r+1/r)}$. If $r$ is lower than the threshold value, then $\ol s^*_1 \leq \ol s'_1$, and otherwise $\ol s^*_1 \geq \ol s'_1$. By solving we find that the threshold value is $2+\sqrt{3}$ (recall that the threshold for WCC was $2$, and that for this value we got the optimal allocation in equilibrium). Clearly, if $r=2+\sqrt{3}$ then $\vec s^*=\vec s'$ and thus $\RPoA(\GG_P,r) = 1$. 

Now, suppose that $r\geq 2+\sqrt{3}$, and thus $\ol s^*_1 \geq  \ol s'_1$. 
Suppose first that $\ol s'_1=0$, i.e. $\frac{2}{a(r+1/r)} \leq n\leq \frac{1}{2a}$. Denote $x=\frac{1}{na}$, then $\frac{r+1/r}{2}\geq  x \geq 2$. 
\begin{align*}
\RPoA(\GG_P,r) & = \frac{n-\frac{2}{a(r+1/r)} + \frac{4}{a(r+1/r)^2}}{n^2 a}  = \frac{1}{an} - \frac{2}{a^2n^2(r+1/r)} + \frac{4}{a^2 n^2 (r+1/r)^2}\\
& = x +  (\frac{4}{(r+1/r)^2} - \frac{2}{r+1/r})x^2 = x+ \frac{4-2(r+1/r)}{(r+1/r)^2}x^2 = x+Ax^2
\end{align*}
for $A=\frac{4-2(r+1/r)}{(r+1/r)^2}$ (note that $A$ is negative).
The expression above is maximized for $x=-\frac{1}{2A}$, which gives us 
$$\RPoA(\GG_P,r) \leq -\frac{1}{4A} = \frac{(r+1/r)^2}{8(r+1/r)-16}.$$

The other case is when $\ol s'_1>0$, and thus $n\geq \frac{1}{2a}$. Note that 
$-\frac{2}{(r+1/r)} + \frac{4}{(r+1/r)^2} \geq -1/4.$
Thus
\begin{align*}
\RPoA(\GG_P,r) & = \frac{n-\frac{2}{a(r+1/r)} + \frac{4}{a(r+1/r)^2}}{n-\frac{1}{4a}} 
\leq \frac{ \frac{1}{2a}-\frac{2}{a(r+1/r)} + \frac{4}{a(r+1/r)^2}}{ \frac{1}{2a}-\frac{1}{4a}} \tag{use lower bound on $n$}\\
& = 4a (\frac{1}{2a}-\frac{2}{a(r+1/r)} + \frac{4}{a(r+1/r)^2}) 
 = 2 - \frac{8}{r+1/r)} + \frac{16}{(r+1/r)^2}.
\end{align*}

It can be verified that this bound is never larger than our previous bound $\frac{(r+1/r)^2}{8(r+1/r)-16}$ (and it is also bounded by a constant).

Thus $\RPoA(\GG_P,r)$ is an increasing function in $r$ (for $r\geq 2+\sqrt{3}$), that is equal to $r/8 + o(r)$. For comparison recall that under WCC costs we had an increase at approximate rate $r/4$.  

For tightness, it is sufficient to look at $\GG_P(a,n)$ for any $a,n$ s.t. $an=\frac{(r+1/r)^2}{2-(r+1/r)}$. Then the inequality we had in the first case becomes an equality.

\medskip
Next, suppose that $r\leq 2+\sqrt{3}$, and thus $\ol s^*_1 \leq \ol  s'_1$.
The first case is when $\ol s^*_1=0, \ol s'_1 = \frac{1}{2a}$. Then $n\leq \frac{2}{a(r+1/r)}$, and
$$
\RPoA(\GG_P,r)  = \frac{n^2 a}{n-\frac{1}{4a}} = \frac{4a^2n^2}{4na-1} = \frac{4x^2}{4x-1}.
$$
The above expression is increasing in $x=na$ in the range $x\leq \frac{2}{r+1/r}$. Thus
$$\RPoA(\GG_P,r) \leq \frac{4 (\frac{2}{r+1/r})^2}{4\frac{2}{r+1/r} - 1} = \frac{16}{8(r+1/r) - (r+1/r)^2}.$$

If $\ol s^*_1>0$, then $\ol s^*_2 = \frac{2}{a(r+1/r)}$ and $na\geq \frac{2}{r+1/r}$. Thus
\begin{align*}
\RPoA(\GG_P,r) &= \frac{n-\frac{2}{a(r+1/r)} + \frac{4}{a(r+1/r)^2}}{n-\frac{1}{4a}} 
 \leq \frac{\frac{2}{a(r+1/r)}-\frac{2}{a(r+1/r)} + \frac{4}{a(r+1/r)^2}}{\frac{2}{a(r+1/r)}-\frac{1}{4a}} \\
& = \frac{4}{a(r+1/r)^2(\frac{2}{a(r+1/r)}-\frac{1}{4a})} = \frac{16}{8(r+1/r) - (r+1/r)^2}.
\end{align*}
That is, we get the same bound in either case. For tightness, it is sufficient to set $a,n$ s.t. $an = \frac{2}{r+1/r}$. 
\end{proof}

\section{Confidence intervals of the Poisson distribution}
\label{apx:poisson}
Suppose that $t$ is a sample from a Poisson distribution with an unknown mean $\ol s'_e$ (the true load). Then for a required confidence level $\alpha$, it holds that $\ol s'_e$ is within the confidence interval $[x_1,x_2]$, where $x_1,x_2$ are the solutions to the equation $(t-x)^2 / t = z$ ($z$ is a constant that depends on $\alpha$). 

Thus $x = t + c/2 \pm \sqrt{c} \sqrt{t+c/4}$. We argue that there is an $r$ s.t. $x_1\cong t/r$ and $x_2\cong t\cdot r$.

 We have
$$x_1 x_2 = t^2 + t c - c(t+c/4) = t^2 - c^2/4 \cong t^2, $$
That is, approximation gets better as $t$ becomes larger.
Thus if we set $r = \sqrt{x_2/x_1}$,
$$(t r)^2 = t^2 \frac{x_2}{x_1} \cong x_1 x_2 \frac{x_2}{x_1} = (x_2)^2,$$
 and 
$$(t/r )^2 = t^2 \frac{x_1}{x_2} \cong x_1 x_2 \frac{x_1}{x_2} = (x_1)^2.$$
So we get that $x_1 \cong t/r, x_2 \cong t r$. 

To see that this is indeed a good approximation, here are the 95\% confidence intervals for various values of $t$:
$$\begin{array}{|c|c||c|c||c|c|}
t & r& x_1 & t/r & x_2 & t r \\
\hline
1	&5.6653111 &	0.1765	&0.1765128&	5.6649&	5.6653111 \\
2	&3.6464040	& 0.5485	&0.5484855&	7.293	&7.2928081\\
3	&2.9403558	& 1.0203	&1.0202846&	8.8212&	8.8210675\\
4	&2.5714977	& 1.5555	&1.5555137&	10.2859&	10.285990\\
5	&2.3411563	& 2.1357	&2.1356967&	11.7058&	11.705781\\
6	&2.1819154	& 2.7499	&2.7498774&	13.0916&	13.091492
\end{array}
$$

From the table we see that the ``correct'' value of $r$ for a given confidence level depends on the load $t$. However recall that people ignore this and consider the standard deviation (and thus $r$) as a fixed fraction from $t$, see Footnote~\ref{fn:KT}.

\end{document}